\documentclass[12pt]{article}
\usepackage{standalone}
\usepackage{amsmath,amssymb, amsthm}
\usepackage{mathrsfs}
\usepackage{indentfirst}
\usepackage{fullpage}
\usepackage{tikz}
\usepackage{bm}
\usepackage{float}
\usepackage{hyperref}
\usepackage{nicefrac}
\usepackage[title]{appendix}
\usepackage{thmtools}
\usepackage{thm-restate}
\usepackage{cleveref}

\usepackage{mleftright}
\mleftright 
\usetikzlibrary{shapes.geometric, shapes.misc, arrows, arrows.meta, decorations.markings, decorations.pathreplacing, calc, positioning}

\definecolor{Xcolor}{HTML}{FF0000}
\definecolor{Ycolor}{HTML}{00FFFF}

\newtheorem{theorem}{Theorem}

\newtheorem*{corollary}{Corollary}
\theoremstyle{definition}

\theoremstyle{remark}

\makeatletter
\newcommand{\subalign}[1]{%
  \vcenter{%
    \Let@ \restore@math@cr \default@tag
    \baselineskip\fontdimen10 \scriptfont\tw@
    \advance\baselineskip\fontdimen12 \scriptfont\tw@
    \lineskip\thr@@\fontdimen8 \scriptfont\thr@@
    \lineskiplimit\lineskip
    \ialign{\hfil$\m@th\scriptstyle##$&$\m@th\scriptstyle{}##$\hfil\crcr
      #1\crcr
    }%
  }%
}
\makeatother

\DeclareMathOperator*{\argmax}{arg\,max}
\DeclareMathOperator*{\argmin}{arg\,min}
\DeclareMathOperator{\diam}{diam}
\DeclareMathOperator{\rad}{rad}

\DeclareMathOperator{\dis}{dis}

\DeclareMathOperator{\tr}{tr}
\DeclareMathOperator*{\defeq}{\buildrel \mathrm{def}\over =}
\DeclareMathOperator{\proj}{proj}

\newcommand{\XX}{\mathbf{X}}
\newcommand{\YY}{\mathbf{Y}}
\newcommand{\FF}{\mathbf{F}}
\newcommand{\GG}{\mathbf{G}}
\newcommand{\VV}{\mathbf{V}} 
\newcommand{\UU}{\mathbf{U}} 
\newcommand{\WW}{\mathbf{W}} 
\newcommand{\R}{\mathbb{R}}
\newcommand{\RR}{\mathbf{R}}
\renewcommand{\SS}{\mathbf{S}}
\newcommand{\DD}{\mathbf{R}}
\newcommand{\HH}{\mathbf{H}}
\newcommand{\KK}{\mathbf{K}}
\newcommand{\RRR}{\mathcal{R}}
\newcommand{\SSS}{\mathcal{S}} 
\newcommand{\dGH}{d_{\mathrm{GH}}}
\newcommand{\nconv}{\mathrm{nconv}}
\newcommand{\lmax}{\lambda_{\max}}
\newcommand{\dmax}{d_{\max}}
\newcommand{\pmax}{p_{\max}}
\newcommand{\obj}{\sigma_c}
\newcommand{\ch}{\mathrm{ch}}
\newcommand{\xx}{\mathbf{x}}
\newcommand{\yy}{\mathbf{y}}
\newcommand{\zz}{\mathbf{z}}
\newcommand{\grad}{\boldsymbol{\nabla}}

\title{Computing the Gromov--Hausdorff distance using gradient methods}
\author{Vladyslav Oles\thanks{\texttt{vlad.oles@proton.me}}}
\date{}
\begin{document}

\maketitle

\begin{abstract}
The Gromov--Hausdorff distance measures the difference in shape between metric spaces and poses a notoriously difficult problem in combinatorial optimization. We introduce its quadratic relaxation over a convex polytope whose solutions provably deliver the Gromov--Hausdorff distance. The optimality guarantee is enabled by the fact that the search space of our approach is not constrained to a generalization of bijections, unlike in other relaxations such as the Gromov--Wasserstein distance.

We suggest conditional gradient descent for solving the relaxation in $O(n^3)$-time per iteration, and demonstrate its performance on metric spaces of hundreds of points. In particular, we use it to obtain a new bound of the Gromov--Hausdorff distance between the unit circle and the unit hemisphere equipped with Euclidean metric. Our approach is implemented as a Python package \texttt{dgh}.
\end{abstract}

\section*{Notation}
\def\arraystretch{1.5}
\vspace{-3mm}
\begin{table}[H]
\begin{tabular}{rp{0.8\textwidth}}
$X \to Y$ & the set of mappings from set $X$ to set $Y$ \\
$d_X$ & metric on set $X$ \\
$\dGH(X, Y)$ & the Gromov--Hausdorff distance between metric spaces $X$ and $Y$\\
 $\mathbf{A}$ & finite matrix with entries $A_{ij}$\\
$c^\mathbf{A}$ & entrywise exponentiation with entries $c^{A_{ij}}$ \\
$\left\|\mathbf{A}\right\|_p$ & entrywise $p$-norm $\left(\sum_{i,j}\left|A_{ij}\right|^p\right)^{1/p}$\\
$\|\mathbf{A}\|_\infty$ & entrywise $\infty$-norm $\max_{i,j}\left|A_{ij}\right|$\\
$\overline{\mathbf{A}\mathbf{B}}$ & line segment $\left\{\alpha\mathbf{A} + (1-\alpha)\mathbf{B}: \alpha \in [0,1]\right\}$\\
\end{tabular}
\end{table}
\def\arraystretch{1}

\section{Introduction}
The Gromov--Hausdroff distance \cite{gromov1999metric} measures the difference in shape between geometric objects, and is a natural metric on the space of (isometry classes of) compact metric spaces. It is widely used in Riemannian geometry and appears in other areas of modern mathematics such as computational topology and graph theory (see e.g. \cite{gromov1999metric, latschev2001vietoris, chazal2009gromov, tuzhilin2020lectures}). Starting from the late 2000s, it has become increasingly popular in data science as a model for dissimilarity measures between shapes such as point clouds \cite{memoli2007use, chazal2009gromov, villar2016polynomial, bronstein2010gromov} and graphs \cite{lee2012persistent, chung2017topological, fehri2018characterizing}).

Formally, the Gromov--Hausdorff distance between a pair of compact metric spaces $X, Y$ minimizes the distortion of distances between the points of $X$ and $Y$ and their images under some bi-directional mapping pair $(f,g)\in(X\to Y)\times(Y\to X)$. Solving this combinatorial minimization is an NP-hard problem, and in fact even approximating it up to a multiplicative factor of $\sqrt{n}$ (where $n \defeq \max\{|X|, |Y|\}$) in the general case remains an intractable task \cite{schmiedl2017computational}. High computational cost of the distance has motivated its relaxations such as the modified Gromov--Hausdorff distance \cite{memoli2012some}, the Gromov--Wasserstein distance (which became a broadly used metric between metric measure spaces in its own right) \cite{memoli2011gromov}, and GHMatch \cite{villar2016polynomial}. Other studies have proposed more feasible algorithms for approximating the Gromov--Hausdorff distance for the special cases of subsets of $\R^1$ \cite{majhi2023approximating}, metric trees \cite{agarwal2018computing, touli2018fpt}, and ultrametric spaces \cite{memoli2021gromov}. 

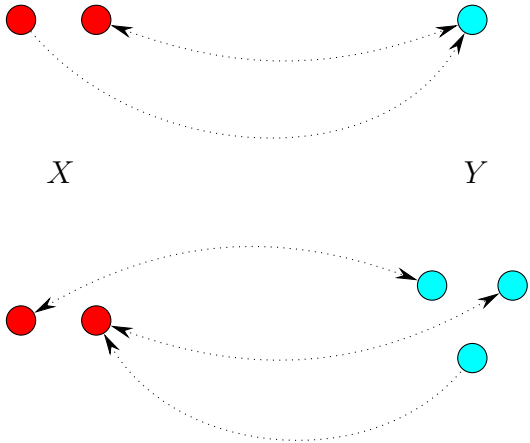
\begin{figure}[h]
    \centering
    \label{fig:nonbijective}
    \caption{An example with $|X| = |Y| = 4$ for which the mapping pairs minimizing the distortion are not bijective. The dotted lines represent one of such mapping pairs. The distances in $X$ and $Y$ take only two distinct values, whose ratio exceeds 2.}
    \vspace{.5cm}
    \begin{tikzpicture}
\node[circle,draw,fill=Xcolor] (X1) at (0,0){};
\node[circle,draw,fill=Xcolor,left of=X1] (X2){};
\node[circle,draw,fill=Xcolor,below of=X1,node distance=4cm] (X3){};
\node[circle,draw,fill=Xcolor,left of=X3] (X4){};
\node[below left=1.6cm and 0cm of X1]{$X$};

\node[circle,draw,fill=Ycolor,right of=X1,node distance=5cm] (Y1){};
\node[circle,draw,fill=Ycolor,below left=3.25cm and .25cm of Y1] (Y2){};
\node[circle,draw,fill=Ycolor,below right=3.25cm and .25cm of Y1] (Y3){};
\node[circle,draw,fill=Ycolor,below of=Y1,node distance=4.5cm] (Y4){};
\node[below right=1.6cm and 4.6cm of X1]{$Y$};

\newcommand{\arrow}{{Stealth[length=3mm, width=1.5mm]}}

\draw[dotted, \arrow-\arrow] (X1) to[out=-20,in=-160] (Y1);
\draw[dotted, -\arrow] (X2) to[out=-50,in=-120] (Y1);

\draw[dotted, \arrow-\arrow] (X4) to[out=30,in=160] (Y2);
\draw[dotted, \arrow-\arrow] (X3) to[out=-20,in=-150] (Y3);
\draw[dotted, \arrow-] (X3) to[out=-60,in=-130] (Y4);

\end{tikzpicture}
\end{figure}

While the Gromov--Wasserstein distance and GHMatch are similar to our proposed approach in leveraging gradient methods to obtain approximate solutions, both of these relaxations are constrained to some generalization of bijections between $X$ and $Y$.
Assuming $|X|=|Y|$ (and, for the Gromov--Wasserstein distance, equal measure assigned to each point), they can retrieve only those solutions $(f,g)$ to the original minimization that satisfy $f = g^{-1}$. However, the existence of such solutions is not guaranteed, as is illustrated by the example from Figure \ref{fig:nonbijective}. In particular, the search space of the Gromov--Wasserstein distance between the depicted $X$ and $Y$ under the uniform probability measure does not contain any of the original solutions. Moreover, the example demonstrates that the minimum distortion over all mapping pairs (which is the ``short'' distance) can be arbitrarily smaller than the minimum distortion over bijections (which is the difference between the ``long'' and the ``short'' distances).
This is perhaps unsurprising considering that $(X\to Y)\times (Y\to X)$ grows super-exponentially faster than the set of bijections between $X$ and $Y$. 
From a practical standpoint, such an expressivity of the Gromov--Hausdorff distance can be useful for aligning metric spaces of irregular density, e.g. those obtained by uneven sampling from the underlying objects.

In the below Section \ref{relaxation}, we propose a parametrized relaxation of the Gromov--Hausdorff distance over the same search space $(X\to Y)\times (Y\to X)$. We derive a threshold value for the parameter ensuring that any solution to the relaxed problem minimizes the distortion and therefore delivers the Gromov--Hausdorff distance. Section \ref{minimization} describes solving the relaxation using conditional gradient descent and discusses the associated computational complexity and optimization landscape. We detail on our implementation and demonstrate its performance in numerical experiments in Section \ref{numerical}. As a byproduct, we tighten the upper bound on the Gromov--Hausdorff distance between a unit circle $S^1 \subset \R^2$ and a unit hemisphere $H^2 \subset \R^3$ (Section \ref{spheres}). For brevity, technical proofs are relegated to the appendix.




\section{Relaxing the Gromov--Hausdorff distance}
\label{relaxation}

\subsection{Matrix reformulation}
Let $X = \{x_1, \ldots, x_n\}$ and $Y = \{y_1, \ldots, y_m\}$ be metric spaces whose finite cardinalities satisfy $n \geq m > 1$, and let $\XX \in \R^{n\times n}$ and $\YY \in \R^{m\times m}$ denote their  distance matrices, e.g. $X_{ij}=d_X(x_i,x_j)$. Recall that the Gromov--Hausdorff distance can be formulated as a combinatorial minimization over the bi-directional mapping pairs $(X\to Y)\times(Y\to X)$
$$\dGH(X, Y) = \frac{1}{2}\min_{\substack{f:X\to Y,\\ g:Y\to X}} \dis\left(\left\{\left(x, f(x)\right): x \in X\right\} \cup \left\{\left(g(y), y\right): y \in Y\right\}\right),\quad\cite{kalton1999distances}$$ where the distortion of some relation $R \subseteq X \times Y$ is defined as the largest absolute difference in distances it incurs: $$\dis R \defeq \max_{\subalign{(x, y),(x', y') \in R}} \left|d_X(x, x') - d_Y(y, y')\right|.$$

For some $f:X\to Y$, consider its ``one-hot encoded'' matrix representation $\FF \in \{0, 1\}^{n \times m}$ s.t. $F_{ij} = \begin{cases}
    1, & f(x_i) = y_j\\
    0, & \text{otherwise}
\end{cases}$ and note that $\left(\FF\YY\FF^T\right)_{ij} = d_Y\left(f(x_i), f(x_j)\right)$. It follows that $$\max_{x,x' \in X} \Big|d_X(x, x') - d_Y\left(f(x), f(x')\right)\Big| = \big\|\XX - \FF\YY\FF^T\big\|_\infty.$$
From the analogous construction of $\GG \in \{0, 1\}^{m \times n}$ for an arbitrary $g: Y \to X$, $$\max_{y,y' \in Y} \left|d_X\left(g(y), g(y')\right) - d_Y(y, y')\right| = \big\|\YY - \GG\XX\GG^T\big\|_\infty.$$
Because $(\XX\GG^T)_{ij} = d_X\left(g(y_j), x_i\right)$ and $(\FF\YY)_{ij} = d_Y\left(f(x_i), y_j\right)$, $$\max_{x \in X, y \in Y} \left|d_X\left(x, g(y)\right) - d_Y\left(f(x), y\right)\right| = \big\|\XX\GG^T - \FF\YY\big\|_\infty.$$

Denote $R = R(f, g) \defeq \left\{\left(x, f(x)\right): x \in X\right\} \cup \left\{\left(g(y), y\right): y \in Y\right\}$. Any pair $(x, y), (x', y') \in R$ satisfies
\begin{align*}
\left|d_X(x, x') - d_Y(y, y')\right| \in \bigg\{&\left|d_X(x, x') - d_Y\left(f(x), f(x')\right)\right|,\\
&\left|d_X\left(g(y), g(y')\right) - d_Y(y, y')\right|,\\
&\left|d_X\left(x, g(y)\right) - d_Y\left(f(x), y\right)\right|,\\
&\left|d_X\left(x', g(y')\right) - d_Y\left(f(x'), y'\right)\right|\bigg\},
\end{align*}
and therefore
\begin{align*}
    \dis R &= \max\bigg\{\begin{aligned}[t]
        &\max_{x,x' \in X} \left|d_X(x, x') - d_Y\left(f(x), f(x')\right)\right|, \\ &\max_{y,y' \in Y} \left|d_X\left(g(y), g(y')\right) - d_Y(y, y')\right|,\\ &\max_{x \in X, y \in Y} \left|d_X\left(x, g(y)\right) - d_Y\left(f(x), y\right)\right|\bigg\}
    \end{aligned}  \\
    &= \left\|\begin{bmatrix}\XX - \FF\YY\FF^T & \XX\GG^T - \FF\YY \\ \GG\XX - \YY\FF^T & \YY - \GG\XX\GG^T\end{bmatrix}\right\|_\infty \\
    &= \big\|\VV - \RR\UU\RR^T + \RR\WW - \WW^T\RR^T\big\|_\infty,
\end{align*}
where $\RR \defeq \begin{bmatrix}\FF & \\ &\GG\end{bmatrix}$ is the matrix representation of $R$ and $\VV \defeq \begin{bmatrix}\XX & \\ &\YY\end{bmatrix}$, $\UU \defeq \begin{bmatrix}\YY & \\ &\XX\end{bmatrix}$, $\WW \defeq \begin{bmatrix}& \YY \\ \XX&\end{bmatrix}$. The presence of a redundant block $\GG\XX-\YY\FF^T$ in the distance difference matrix $\VV - \RR\UU\RR^T + \RR\WW - \WW^T\RR^T$ is motivated by matrix symmetry.

By construction, $\RR \in \{0, 1\}^{(n+m)\times(n+m)}$ is row-stochastic and has $m\times m$ and $n\times n$ blocks of zeros in the upper right and lower left, respectively. Let $\RRR \subset \{0, 1\}^{(n+m)\times(n+m)}$ denote the set of all such matrices, which is in a 1-to-1 correspondence with the bi-directional mapping pairs $(X\to Y)\times(Y\to X)$. We will write $\dis \RR \defeq \left\|\VV - \RR\UU\RR^T + \RR\WW - \WW^T\RR^T\right\|_\infty$ to denote the matrix-based formulation of distortion, and assume that the distinction between $\dis:\RRR\to \R$ and $\dis:\mathscr{P}(X\times Y)\to\R$ is clear from the context. Then an equivalent formulation of the Gromov--Hausdorff distance can be given as
\begin{equation*}
    \label{eqn:matrix_reformulation}
    \dGH(X, Y) = \frac{1}{2}\min_{\RR \in \RRR} \dis \RR.
    \tag{$\star$}
\end{equation*}

\subsection{Relaxing the objective}
The $\infty$-norm in $\dis \RR$ deprives the (otherwise quadratic) objective of (\ref{eqn:matrix_reformulation}) of differentiability. A standard trope in smooth relaxations of the maximum function is to involve the sum of exponents of its arguments. It turns out that, for a sufficiently large exponentiation base, this yields a smooth relaxation of (\ref{eqn:matrix_reformulation}) that is minimized only by solutions to (\ref{eqn:matrix_reformulation}). By construction, such a relaxation favors solutions aligning the distances between $X$ and $Y$ better on average (in addition to minimizing their largest difference), which may be desirable in applications.

Let $\Delta \defeq \left\{\left|d_X(x, x') - d_Y(y, y')\right|: x, x' \in X, y, y' \in Y\right\}$, and note that $\dis \RR \in \Delta$ for any $\RR \in \RRR$.
The \textit{distortion gap} between $X$ and $Y$ 
is then defined as $$\rho = \rho(X, Y) \defeq \min \left\{|\delta - \delta'|: \delta, \delta' \in \Delta, \delta \neq \delta'\right\}.$$ Trivially, $\dis \RR \neq \dis \RR'$ for some $\RR, \RR' \in \RRR$ implies that $|\dis \RR - \dis \RR'| \geq \rho$, which provides justification for the name.

\begin{restatable}{theorem}{cthreshold}
    \label{thm:c_threshold}
    Let $c \geq \left(\frac{(n+m)^2-n-m}{2}\right)^{1/\rho}$. Then
    \begin{align*}
    \argmin_{\RR\in\RRR}\left\|c^{\VV - \RR\UU\RR^T + \RR\WW - \WW^T\RR^T} + c^{\RR\UU\RR^T - \VV + \WW^T\RR^T - \RR\WW}\right\|_1 \subseteq \argmin_{\RR\in\RRR}\dis\RR,
    \end{align*}
    where the exponentials are taken entrywise.

\end{restatable}

Theorem \ref{thm:c_threshold} is based on the idea that a decrease in $\left\|\VV - \RR\UU\RR^T + \RR\WW - \WW^T\RR^T\right\|_\infty$, the largest magnitude in the distance difference matrix, must also mean a decrease in its relaxation $\left\|c^{\VV - \RR\UU\RR^T + \RR\WW - \WW^T\RR^T} + c^{\RR\UU\RR^T - \VV + \WW^T\RR^T - \RR\WW}\right\|_1$ even when the magnitudes of all other distance differences increase from zero to the new maximum. At the same time, there is little reason to expect better distance alignment from suboptimal choices of $\RR \in \RRR$ than from solutions to (\ref{eqn:matrix_reformulation}),
and in practice much smaller values of $c$ than $\left(\frac{(n+m)^2-n-m}{2}\right)^{1/\rho}$ can satisfy the statement of Theorem \ref{thm:c_threshold}.

Note that $c$ and $c^{-1}$ behave identically in the 1-norm relaxation, which means that both $c \in (0, 1]$ and $c \in [1, \infty)$ can be considered for analogous results. For simplicity, we focus on the latter option and assume $c \geq 1$ throughout this work.

Recall that the two components of the distance difference matrix, $\VV - \RR\UU\RR^T$ and $\RR\WW - \WW^T\RR^T$, have complementary block sparsity: the former contains $n\times m$ zeros in the upper right and $m\times n$ zeros in the lower left, while the latter has $n\times n$ zeros in the upper left and $m\times m$ zeros in the lower right. It follows that
\begin{align*}
    &\left\|c^{\VV - \RR\UU\RR^T + \RR\WW - \WW^T\RR^T} + c^{\RR\UU\RR^T - \VV + \WW^T\RR^T - \RR\WW}\right\|_1 \\ &\hspace{2cm}= \left\|c^{\VV - \RR\UU\RR^T} + c^{\RR\WW - \WW^T\RR^T} - c^\mathbf{0} + c^{\RR\UU\RR^T - \VV} + c^{\WW^T\RR^T - \RR\WW} - c^\mathbf{0}\right\|_1 \\ &\hspace{2cm}= \left\|c^{\VV - \RR\UU\RR^T}\right\|_1 + \left\|c^{\RR\UU\RR^T - \VV}\right\|_1 + \left\|c^{\RR\WW - \WW^T\RR^T}\right\|_1 + \left\|c^{\WW^T\RR^T - \RR\WW}\right\|_1 - 2(n+m)^2.
\end{align*}

Leveraging the structure of $\RR$ and subsequently applying the trace trick gives
\begin{align*}
&\left\|c^{\VV - \RR\UU\RR^T}\right\|_1 = \left\langle c^{\VV}, c^{-\RR\UU\RR^T}\right\rangle = \left\langle c^{\VV}, \RR c^{-\UU}\RR^T\right\rangle = \left\langle \RR , c^{\VV}\RR c^{-\UU}\right\rangle
\end{align*}
and
\begin{align*}
&\left\|c^{\RR\WW - \WW^T\RR^T}\right\|_1 = \left\langle c^{\RR\WW}, c^{-\WW^T\RR^T}\right\rangle = \left\langle \RR c^{\WW},
\big(\RR c^{-\WW}\big)^T\right\rangle = \left\langle \RR , \big(c^{\WW}\RR c^{-\WW}\big)^T\right\rangle,
\end{align*}
as well as $\left\|c^{\RR\UU\RR^T - \VV}\right\|_1 = \Big\langle \RR, c^{-\VV}\RR c^{\UU}\Big\rangle$ and $\left\|c^{\WW^T\RR^T - \RR\WW}\right\|_1 = \left\langle \RR , \left(c^{-\WW}\RR c^{\WW}\right)^T\right\rangle$. Combining the four equations casts the 1-norm relaxation of (\ref{eqn:matrix_reformulation}) as a quadratic minimization
\begin{align*}
    \label{eqn:quadratic_objective}
    \min_{\RR \in \RRR} \obj(\RR) \defeq \left\langle \RR, c^{\VV}\RR c^{-\UU} + c^{-\VV}\RR c^{\UU} + \big(c^{\WW}\RR c^{-\WW} + c^{-\WW}\RR c^{\WW}\big)^T\right\rangle,
    \tag{$\star\star$}
\end{align*}
which concludes the proof of the following
\begin{corollary}[to Theorem \ref{thm:c_threshold}]
    Let $c \geq \left(\frac{(n+m)^2-n-m}{2}\right)^{1/\rho}$. Then
    \begin{align*}
    \argmin_{\RR\in\RRR} \obj(\RR) \subseteq \argmin_{\RR\in\RRR}\dis\RR.
    \end{align*}
\end{corollary}

\subsection{Relaxing the domain}
In order to enable first-order methods for solving (\ref{eqn:quadratic_objective}), its objective $\obj$ needs to be considered over a continuous domain. A common approach in combinatorial optimization is relaxing the discrete domain to its convex hull, employing a gradient-based algorithm on the convex region, and projecting the resulting solution back onto the original domain.
The convex hull of $\RRR$, henceworth denoted as $\SSS$, is the set of all $(n+m)\times(n+m)$ row-stochastic matrices with $m\times m$ and $n\times n$ blocks of zeros in the upper right and lower left, respectively.
This is a direct consequence of generalizing the Birkhoff--von Neumann theorem to the row-stochastic matrices, which was done e.g. in \cite{gubin2008subgraph} and \cite{cao2022centrosymmetric}. Because the points of $\SSS$ represent bi-directional pairs of generalized (or ``soft'') mappings, we refer to it as the \textit{bi-mapping polytope}. Searching over the bi-mapping polytope $\SSS$ is a key difference between our approach and other relaxations of the Gromov--Hausdorff distance, which are constrained to the Birkhoff polytope. 

We will now show that projecting a solution to the resulting continuous relaxation
\begin{equation*}
    \label{eqn:continuous_optimization}
    \min_{\SS \in \SSS} \obj(\SS) \defeq \left\langle \SS, c^{\VV}\SS c^{-\UU} + c^{-\VV}\SS c^{\UU} + \big(c^{\WW}\SS c^{-\WW} + c^{-\WW}\SS c^{\WW}\big)^T\right\rangle
    \tag{$\star\star$$\star$}
\end{equation*}
back onto the vertices $\RRR$ always yields a solution to (\ref{eqn:quadratic_objective}). Observe that the faces of $\SSS$ can be characterized similarly to those of the $(n+m)$-th Birkhoff polytope, see e.g. \cite{paffenholz2013faces}. Specifically, every face of $\SSS$ corresponds to a (possibly empty) set of forbidden assignments between the points of $X$ and $Y$ and is comprised of the convex combinations of all the compliant mapping pairs. We write it formally as
\begin{restatable}{lemma}{polytopegeometry}
    \label{lem:polytope_geometry}
    Any face $\Phi$ of the bi-mapping polytope $\SSS$ is characterized by an index set $(\mathcal{I}, \mathcal{J}) \subset \{1, \ldots, n+m\}\times\{1,\ldots,n+m\}$ s.t. $\Phi = \left\{\SS \in \SSS: S_{ij} = 0 \quad\forall (i,j)\in (\mathcal{I}, \mathcal{J})\right\}$.
\end{restatable}
We can now show that any face of $\SSS$ containing a solution on its interior must be a part of the solution set, which in particular would imply $\argmin_{\RR \in \RRR}\obj(\RR) \subseteq \argmin_{\SS \in \SSS}\obj(\SS)$.
\begin{restatable}{lemma}{solutiononinterior}
\label{lem:solution_on_interior}
Let $\SS^* \in \argmin_{\SS \in \SSS}\obj(\SS)$. If $\SS^* \in \Phi \setminus \partial\Phi$ for some face $\Phi$, then $$\Phi \subseteq \argmin_{\SS \in \SSS}\obj(\SS).$$
\end{restatable}

Moreover, the characterization from Lemma \ref{lem:polytope_geometry} can be used to establish that the nearest vertex (or vertices) to any point in $\SSS$ must belong to the same faces as the point itself. Let $\proj_\RRR: \SSS \to \RRR$ denote some projection of the bi-mapping polytope onto the set of its vertices (note that the ambiguity in defining $\proj_\RRR \SS \in \argmin_{\RR \in \RRR}\|\SS - \RR\|_2$ is due to the non-convexity of $\RRR$).
\begin{restatable}{lemma}{voronoiface}
    \label{lem:voronoi_face}
    Let $\Phi$ be a face of $\SSS$. If $\SS \in \Phi$, then $\proj_\RRR \SS \in \Phi$.
\end{restatable}
It follows from Lemmas \ref{lem:solution_on_interior} and \ref{lem:voronoi_face} that any nearest vertex to a solution must be a solution itself. Combining the above with the corollary to Theorem \ref{thm:c_threshold} concludes the proof of
\begin{theorem}
    \label{thm:guarantees}
    If $c \geq \left(\frac{(n+m)^2-n-m}{2}\right)^{1/\rho}$, then $$\proj_\RRR \left(\argmin_{\SS\in\SSS} \obj(\SS)\right) \subseteq \argmin_{\RR\in\RRR}\dis\RR.$$
\end{theorem}

Theorem \ref{thm:guarantees} provides theoretical guarantees for the solutions to (\ref{eqn:continuous_optimization}) to deliver the Gromov--Hausdorff distance, and is the main theoretical result of this work.

\section{Solving the relaxation}
\label{minimization}
\subsection{Conditional gradient descent}
(\ref{eqn:continuous_optimization}) is an indefinite 
quadratic minimization with affine constraints and a Lipschitz gradient $$\grad\obj(\SS) = 2\left(c^{\VV}\SS c^{-\UU} + c^{-\VV}\SS c^{\UU} + \big(c^{\WW}\SS c^{-\WW} + c^{-\WW}\SS c^{\WW}\big)^T\right).$$ While finding its global minimum remains an NP-hard problem, approximate solutions can be efficiently obtained by the Frank--Wolfe algorithm \cite{frank1956algorithm}, also known as conditional gradient descent.

The iterative algorithm starts at some $\SS_0 \in \SSS$. At every iteration, it finds the descent direction as a point in $\SSS$ that minimizes the cosine similarity with the gradient at the current point $\SS_i$,
$$\DD_i \in \argmin_{\SS \in \SSS}\left\langle \SS, \grad\obj(\SS_i)\right\rangle.$$
The descent direction $\DD_i$ chosen by the algorithm is always a vertex of $\SSS$ (note that the minimum must be attained at some vertex due to the linearity of the problem).
The algorithm then finds a point on the line segment $\overline{\SS_i\DD_i}$ that minimizes $\obj$, $$\SS_{i+1} \in \argmin_{\gamma \in [0, 1]} \obj\left(\gamma\SS_i + (1-\gamma)\DD_i\right),$$
which concludes the $i$-th iteration.

The algorithm's convergence is measured as the \textit{Frank--Wolfe gap} $\left\langle \SS_i - \DD_i, \grad \obj(\SS_i) \right\rangle \geq 0$, which is zero if and only if $\SS_i$ is a stationary point. 
The algorithm terminates when the Frank--Wolfe gap becomes sufficiently small (or after reaching the iteration limit). It takes the Frank--Wolfe algorithm $O(\epsilon^{-2})$ iterations to approach a stationary point with the gap of $\epsilon$ \cite{lacoste2016convergence}, each iteration requiring $O(n^3)$ time.

\subsection{Stationary points}
The structure of $\SSS$ helps characterize the stationary points of (\ref{eqn:continuous_optimization}). The following result suggests that the trend of saddle point prevalence in high-dimensional non-convex optimization \cite{dauphin2014identifying} is unlikely to manifest in (\ref{eqn:continuous_optimization}) for a broad class of metric spaces (e.g. random metric spaces from \cite{vershik2004random}).
\begin{restatable}{theorem}{oneminimizer}
    \label{thm:one_minimizer}
    Let the distances stored in $\XX$ be realized by continuous random variables $\mathrm{D}_1, \ldots, \mathrm{D}_{\frac{n(n-1)}{2}}$ such that any $\mathrm{D}_i$ restricted to any permissible realization of the rest of the variables $\left\{\mathrm{D}_j=d_j:j\neq i\right\}$ has support of non-zero measure. If $c > 1$, then the linear minimization from the first-order necessary optimality condition for any $\SS^*\in\SSS$ almost surely has a single solution:
    $$\mathbb{P}_\XX\left[\Big|\argmin_{\SS\in\SSS} \left\langle \SS, \grad\obj(\SS^*) \right\rangle\Big| > 1\right] = 0.$$
\end{restatable}


For any $\SS^* \in \SSS$, the probability distribution of the distances in $\XX$ induces a random matrix $\grad\obj(\SS^*)$ and a Bernoulli variable representing whether $\SS^*$ is a stationary point of (\ref{eqn:continuous_optimization}). Under the assumption of independence of these two random objects, Theorem \ref{thm:one_minimizer} entails that any stationary point $\SS^*$ almost surely satisfies its second-order sufficient optimality condition as the latter ends up being imposed on the empty set $\argmin _{\SS\in\SSS} \left\langle \SS, \grad\obj(\SS^*) \right\rangle\setminus \left\{\SS^*\right\}.$ In practice, it means that the approximate solutions to (\ref{eqn:continuous_optimization}) recovered by the Frank--Wolfe algorithm are likely to be points of local minima and not saddle points.

\subsection{Non-convexity and the choice of $c$}
\label{convexity}
The Hessian of $\obj$ is a constant $(n+m)^2\times(n+m)^2$ matrix $$\HH \defeq \grad^2\obj(\cdot) = 2\left(c^{\VV}\otimes c^{-\UU} + c^{-\VV}\otimes c^{\UU} + \left(c^{\WW}\otimes c^{-\WW^T} + c^{-\WW}\otimes c^{\WW^T}\right)\KK\right),$$ where $\otimes$ is the Kronecker product and $\KK \in \{0, 1\}^{(n+m)^2\times(n+m)^2}$ denotes the commutation matrix. 
The standard approach to navigating the non-convexity of $\obj$ is to consider the space spanned by the eigenvectors corresponding to negative eigenvalues of $\HH$. However, doing so has $O(n^6)$-time and $O(n^4)$-memory complexity and is impractical for most cases. As a cheaper substitute, we assess the extent of non-convexity of $\obj$ using the normalized nuclear norm-induced distance from $\HH$ to the set of positive semidefinite matrices $\mathcal{M}^+$: $$\nconv(\obj) \defeq \inf_{\mathbf{M} \in \mathcal{M}^+} \frac{\|\HH - \mathbf{M}\|_*}{\|\HH\|_*} = \frac{\lambda^-}{\lambda^+ + \lambda^-} \in [0, 1],$$ where $\|\cdot\|_*$ denotes the nuclear norm and $\lambda^-$ and $\lambda^+$ are the total magnitudes of respectively negative and positive eigenvalues of $\HH$ \cite{davydov2019searching}. In particular, this distance is equal to 0 (or 1) if and only if the function is convex (or concave).

Note that $\lambda^+ - \lambda^- = \tr\HH = 8(n+m)^2$ due to the zero main diagonals of $\VV$ and $\UU$ and
\begin{align*}
    \tr\left(\left(c^\WW \otimes c^{-\WW^T}\right)\KK\right) = \left\langle \KK,  c^\WW \otimes c^{-\WW^T} \right\rangle 
    = \sum_{i,j=1}^{n+m} c^{W_{ij}}c^{(-\WW^T)_{ij}}
    = (n+m)^2
    = \tr\left(\left(c^{-\WW} \otimes c^{\WW^T}\right)\KK\right),
\end{align*}
and therefore $$\nconv(\obj) = \frac{\lambda^+ - 8(n+m)^2}{2\lambda^+ - 8(n+m)^2}\in [0, \frac{1}{2}).$$
Let $\lmax > 0$ denote the dominant eigenvalue of $\HH$ as per the Perron--Frobenius theorem. In the degenerate case of $c=1$ this is the only non-zero eigenvalue of the rank-1 Hessian $\HH=8\left(\mathbf{1}\mathbf{1}^T\right)$, which makes $\obj$ convex (though uninformative). 
On the other hand, the Courant--Fischer theorem yields $$\lmax \geq \frac{\mathbf{1}^T\HH\mathbf{1}}{\mathbf{1}^T\mathbf{1}} = \frac{\|\HH\|_1}{(n+m)^2} \xrightarrow[c \to \infty]{} \infty,$$
and
$\lambda^+ \geq \lmax$
then entails
$\nconv(\obj) \xrightarrow[c \to \infty]{} \frac{1}{2}.$
We conclude that the landscape of (\ref{eqn:continuous_optimization}) starts flat at $c=1$ and becomes increasingly non-convex as $c$ grows. In particular, the following result provides a tractable ($O(n^2)$-time and -memory) bound on its non-convexity based on the value of $c$.
\begin{restatable}{theorem}{nonconvexity}
    \label{thm:nonconvexity}
    Let $\alpha \in \left[0, \frac{1}{2}\right)$ and $c \geq 1$ satisfy 
    $$\frac{2\alpha}{\frac{1}{2}-\alpha} = \frac{\sqrt{16(n+m)^4 + \left(c^{\dmax} + c^{-\dmax} - 2\right)p_{\max} - \frac{16}{(n+m)^4}\left\|c^{\WW}\right\|_1^2\left\|c^{-\WW}\right\|_1^2}}{n+m},$$
    where $\dmax \defeq \max\{\diam X, \diam Y\}$ 
    and $p_{\max} \defeq \frac{\left(2\sqrt{2} + 4\right)(n+m)^2\sqrt{(n+m)^2\|\WW\|_2^2 - \|\WW\|_1^2}}{\dmax}+6$. Then $\nconv(\obj) \leq \alpha$.
\end{restatable}
The choice of $c$ is therefore a balancing act between increasing the fraction of solutions (i.e. global minima) delivering the Gromov--Hausdorff distance and decreasing the probability of discovering a solution by the Frank--Wolfe algorithm. As $c\to\infty$, the fraction stops growing after attaining 1 at some $c = c^*$, while the chances of discovering a solution continue to decay due to the increasingly non-convex optimization landscape. Because Theorem \ref{thm:c_threshold} implies $c^* \leq \left(\frac{(n+m)^2-n-m}{2}\right)^{1/\rho}$, it is never optimal to choose $c$ above this bound.

\section{Numerical experiments}
\label{numerical}
We demonstrate the method's performance using our implementation from a Python package \texttt{dgh}.
Given a choice of $c$ and a budget of Frank--Wolfe iterations, it starts from a random $\SS_0 \in \SSS$ and iterates until the Frank--Wolfe gap is at most $10^{-16}$ (or there are no iterations left), after which the last $\SS_i$ is projected to the nearest $\RR \in \RRR$. The cycle repeats until the budget is depleted, after which the smallest found $\frac{1}{2}\dis \RR$ is returned as an estimate and upper bound of the Gromov--Hausdoff distance (we note that this approach is easily parallelizable as there is no interdependency between the random restarts).
In a bid to save on redundant computations, \texttt{dgh} also compares every new best $\frac{1}{2}\dis \RR$ against the trivial lower bound $$\dGH(X, Y) \geq \frac{1}{2}\max\left\{\left|\diam X - \diam Y\right|, \left|\rad X - \rad Y\right|\right\}\quad\cite{memoli2012some}$$ and terminates whenever they match, which means that the Gromov--Hausdorff distance was found exactly. Prior to the computations, \texttt{dgh} normalizes all distances in the input metric spaces so that $\max\{\diam X, \diam Y\} = 1$ to avoid floating-point arithmetic overflow, and scales the resulting $\dis\RR$ back afterwards. 

In the following, we describe our numerical experiments on synthetic and model metric spaces. 
All computations were performed on a standard 2016 Intel i7-7500U processor.



\subsection{Synthetic metric spaces}

We evaluate the speed and accuracy of \texttt{dgh} on synthetic point clouds and metric graphs by computing the Gromov--Hausdorff distance from each space to its isometric copy. A point cloud is generated by uniformly sampling $n=200$ points from the unit cube in $\R^3$ and taking the Euclidean distance between them. A graph is repeatedly generated according to the Erd\H{o}s-R\'{e}nyi model with $n=200$ vertices and the edge probability of $p=0.05$ until it is connected, and then endowed with the shortest path metric. We generate 100 point clouds and 100 metric graphs, and run the experiment on each metric space for $c=1+10^{-10},1+10^{-9},\ldots,1+10^{40}$ and with the budgets of 100 and 1000 iterations. Note that (the matrix representation of) a mapping pair $(f, g)$ is a solution to (\ref{eqn:quadratic_objective}) if and only if $f=g^{-1}$ is an isometry, which corresponds to an all-zeros distance difference matrix. As a consequence, the solution set of the relaxation coincides with the original solutions to the Gromov--Hausdorff distance (\ref{eqn:matrix_reformulation}) for any $c>1$. 




\begin{figure}[h]
    \centering
    \caption{Performance of \texttt{dgh} on synthetic metric spaces with $n=200$ points.} 
    \vspace{.25cm}
    \includegraphics[width=17cm]{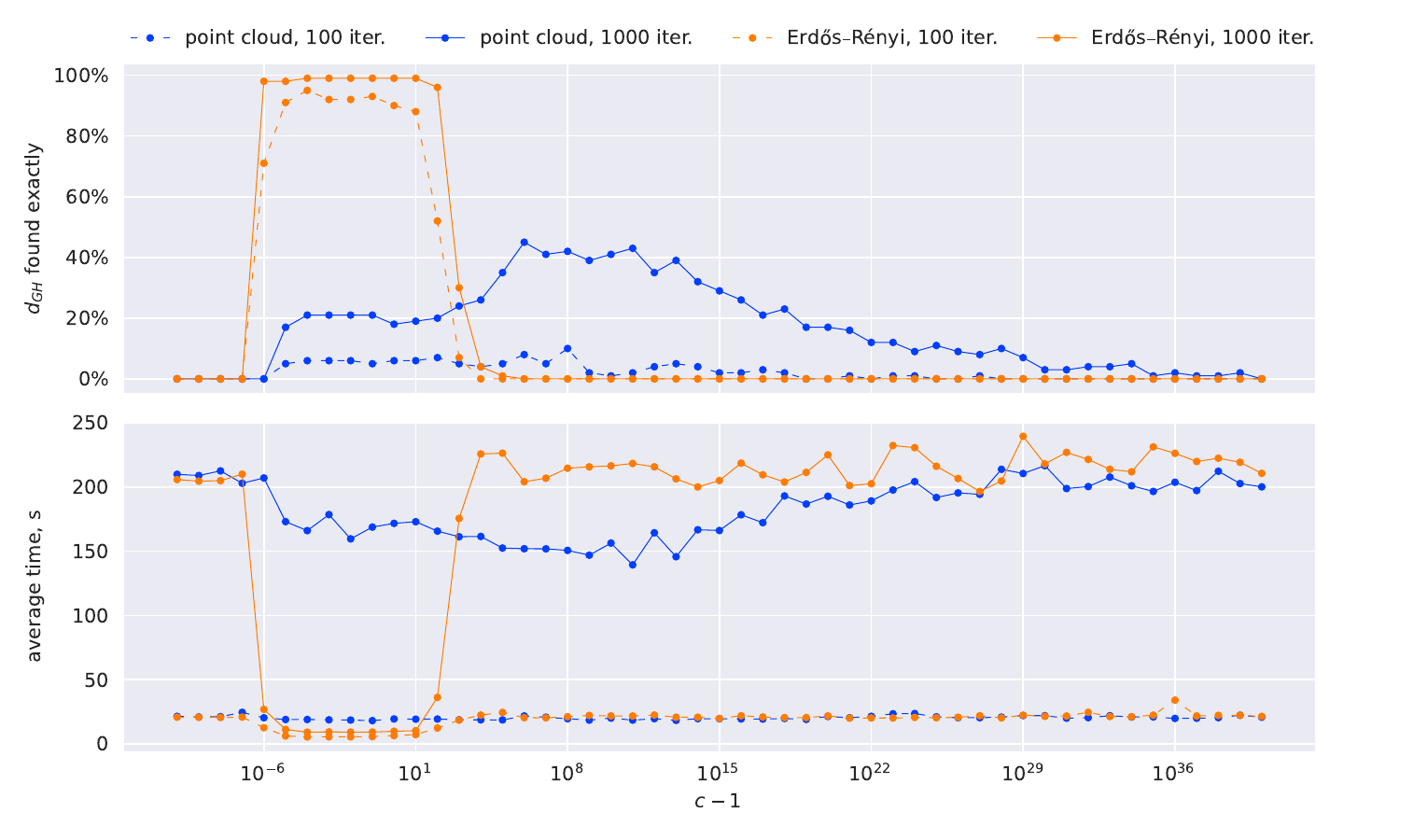}
    \label{fig:performance}
\end{figure}

Figure \ref{fig:performance} shows the percentage of experiments in which the Gromov--Hausdorff distance was found exactly and the average compute time required by \texttt{dgh}, separately for each metric space type and iteration budget combination. 
Predictably, the accuracy of the algorithm on either metric space type decreases once $c$ becomes large enough. The initial positive trend may be explained by the associated reduction of maximal distance misalignment in local minima on $\SSS$ bringing them closer to isometries, which in turn increases the chance to find one by projecting an approximate solution of (\ref{eqn:continuous_optimization}) onto $\RRR$. Similarly to the logic underlying Theorem \ref{thm:c_threshold}, we expect the scale of $c$ on which the above effect is manifested to be inversely proportional to the distortion gap. Note that the distortion gap on metric graphs, $\rho=1$ due to their integer distances, is orders of magnitude bigger than on point clouds even after rescaling by their diameters --- $O(\log{n})$ for the Erd\H{o}s-R\'{e}nyi graph and $\sqrt{3}$ for the unit cube. For the rescaled inputs, the average distortion gap is $0.24$ in metric graphs and $4.6 \times 10^{-17}$ in point clouds. This difference likely accounts for the contrasting optimal scales of the parameter $c$ in each case.

\subsection{Model spaces}
\label{spheres}
Recall that for any compact metric space $X$ there exists its finite $\varepsilon$-net $X_\varepsilon \subseteq X$ for any $\varepsilon > 0$. Because $$\left|\dGH(X,Y) - \dGH(X_\varepsilon,Y_\varepsilon)\right| \leq \epsilon,\quad\cite{oles2022lipschitz}$$ such an approximation enables numerical estimation of the Gromov--Hausdorff distance between infinite $X,Y$ to an arbitrary precision. This can be particularly useful for estimating the distance between model metric spaces. We demonstrate it by refining an upper bound of the Gromov--Hausdorff distance between the unit circle $S^1 \subset \R^2$ and the unit hemisphere $H^2 \subset \R^3$ (hemisphere of the unit sphere), established to be strictly below $\frac{\sqrt{3}}{2}$ in \cite{lim2021gromov}.

To generate an $\varepsilon$-net of $H^2$, we use a slight modification of the regular construction described in \cite{deserno2004generate}.
Given some small $\delta$, we consider evenly spaced polar angles $\theta_i$ covering the hemisphere range $\left[0, \frac{\pi}{2}\right]$ so that the geodesic distance from any $\mathbf{x} = (\theta, \phi) \in H^2$ to the nearest $\yy = (\theta_i, \phi)$ is $|\theta - \theta_i| \leq \frac{\delta}{2}$. The law of cosines then bounds the Eucledian distance between $\xx$ and $\yy$ by
\begin{align*}
    \|\xx - \yy\|^2 &= 2 - 2\cos(\theta - \theta_i)\\
    &\leq 2\left(1-\cos\frac{\delta}{2}\right).
\end{align*}
For each $\theta_i$, we choose evenly spaced azimuthal angles $\phi_j(\theta_i)$ so that 
the geodesic distance between any $\mathbf{y} = (\theta_i, \phi)$ and the nearest $\mathbf{z} = \left(\theta_i, \phi_j(\theta_i)\right)$ is at most $\frac{\delta}{2}$. Because $\yy$ and $\zz$ are on a circle of radius $\sin\theta_i$, this implies $\left|\phi - \phi_j(\theta_i)\right| \leq \frac{\delta}{2\sin\theta_i}$ and therefore 
\begin{align*}
    \|\yy - \zz\|^2 &= 2\sin^2\theta_i - 2\sin^2\theta_i \cos\left(\phi - \phi_j(\theta_i)\right) \\
    &\leq 2\sin^2\theta_i\left(1-\cos\frac{\delta}{2\sin\theta_i}\right).
\end{align*}
The $\varepsilon$-net $H^2_\varepsilon$ is comprised of the points at $\left(\theta_i, \phi_j(\theta_i)\right)$ for every $i,j$ pair. Its covering radius $\varepsilon$ can be bounded based on
\begin{align*}
    \left|\cos\angle\xx\yy\zz\right| &= \left| \langle \xx - \yy, \yy - \zz\rangle\right| \\
    &= \sin\theta_i\left|\sin\theta_i - \sin\theta\right|\left(\cos\left(\phi - \phi_j(\theta_i)\right)\right) \\
    &\leq \sin\theta_i\left(\sin\theta_i\left(1-\cos\frac{\delta}{2}\right)+\cos\theta_i\sin\frac{\delta}{2}\right)\left(1 - \cos\frac{\delta}{2\sin\theta_i}\right),
\end{align*}
which entails
\begin{align*}
    \varepsilon^2 &\leq \sup_{\xx}\|\xx-\zz\|^2 \\
    &\leq \sup_{\xx} \left(\|\xx - \yy\|^2 + \|\yy - \zz\|^2 + 2\|\xx - \yy\|\|\yy - \zz\|\left|\cos{\angle \xx\yy\zz}\right|\right) \\
    &\leq 2\left(1-\cos\frac{\delta}{2}\right) + 2 \max_i \sin^2\theta_i \left(1 - \cos\frac{\delta}{2\sin\theta_i}\right)\Bigg[1 + \\&\hspace{2cm}2\sqrt{1-\cos\frac{\delta}{2}}\sqrt{1 - \cos\frac{\delta}{2\sin\theta_i}}\left(\sin\theta_i\left(1-\cos\frac{\delta}{2}\right)+\cos\theta_i\sin\frac{\delta}{2}\right)\Bigg].
\end{align*}

\begin{figure}[h]
    \centering
    \caption{The relation between $S^1_\varepsilon$ and $H^2_\varepsilon$ induced by a mapping pair recovered by \texttt{dgh}.}
    \vspace{.25cm}
    \includegraphics[width=17cm]{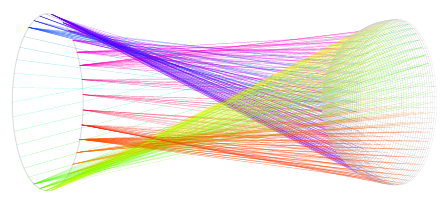}
    \label{fig:spheres}
\end{figure}

By letting $\delta=\frac{\pi}{13}$ in the above, we construct $H^2_\varepsilon$ of 437 points and $\varepsilon \approx 0.0854$. We match this covering radius on $S^1$ by constructing its $\varepsilon$-net $S^1_\varepsilon$ as a regular lattice of 37 points. Setting $c=10^{12}$ 
and running \texttt{dgh} with a permissive iteration budget for about 7 minutes (1255 iterations) yields a mapping pair delivering $\dGH(S^1_\varepsilon, H^2_\varepsilon) < 0.7009$
, shown in Figure \ref{fig:spheres}; we note that neither of the recovered mappings is injective. It follows that $$\dGH\left(S^1, H^2\right) \leq \dGH\left(S^1_\varepsilon, H^2_\varepsilon\right) + \varepsilon < 0.7863,$$ a refinement over $\frac{\sqrt{3}}{2} \approx 0.8660.$ Trivially, this bound holds for both the open and the closed hemispheres as well as for the ``helmet'' of $S^2$ --- its upper hemisphere that contains $(0, \phi) \in S^2$ if and only if $\phi \in \left[0, \pi\right)$, a construction from \cite{lim2021gromov} facilitating antipode-preserving mappings.



\section*{Acknowledgements}
I am grateful to Oleksandr Dykhovychniy, Kostiantyn Lyman, and Kevin R. Vixie for the numerous insightful conversations that helped shape this paper. I would also like to acknowledge the convenience of an online tool for computing matrix derivatives \cite{laue2018computing} used in this study.

\begin{appendices}
\section{Proofs}
\label{proofs}
\cthreshold*
\begin{proof}
    Recall that $\VV - \RR\UU\RR^T + \RR\WW - \WW^T\RR^T$ is by construction symmetric with zero main diagonal, and $\dis \RR$ is the largest magnitude among its entries. Let     $\RR, \RR' \in \RRR$ be such that $\dis \RR > \dis \RR'$ and therefore $\dis \RR \geq \dis\RR' + \rho$.
    Because $a \mapsto c^a + c^{-a}$ is convex and attains its minimum at $a = 0$,
    \begin{align*}
        \left\|c^{\VV - \RR\UU\RR^T + \RR\WW - \WW^T\RR^T} + c^{\RR\UU\RR^T - \VV + \WW^T\RR^T - \RR\WW}\right\|_1 &\geq 2\left(c^{\dis \RR} + c^{-\dis \RR} + (n+m)^2-2\right) \\
        &> 2\left(c^{\dis \RR} + (n+m)^2 - 2\right).
    \end{align*}
    At the same time,
    \begin{align*}
        \hspace{-.5cm}\left\|c^{\VV - \RR'\UU\RR'^T + \RR'\WW - \WW^T\RR'^T} + c^{\RR'\UU\RR'^T - \VV + \WW^T\RR'^T - \RR'\WW}\right\|_1 &\leq \left((n+m)^2 - n - m\right)\left(c^{\dis \RR'} + c^{-\dis \RR'}\right) + 2(n+m) \\ &< \left((n+m)^2-n-m\right)c^{\dis \RR'} + (n+m)^2+n+m.
    \end{align*}
    Then
    \begin{align*}
        &\left\|c^{\VV - \RR\UU\RR^T + \RR\WW - \WW^T\RR^T} + c^{\RR\UU\RR^T - \VV + \WW^T\RR^T - \RR\WW}\right\|_1 - \left\|c^{\VV - \RR'\UU\RR'^T + \RR'\WW - \WW^T\RR'^T} + c^{\RR'\UU\RR'^T - \VV + \WW^T\RR'^T - \RR'\WW}\right\|_1 \\
        &\quad> 2c^{\dis \RR} - \left((n+m)^2-n-m\right)c^{\dis \RR'} + (n+m)^2-n-m-4 \\
        &\quad\geq \left(2c^\rho - (n+m)^2-n-m\right)c^{\dis \RR'}  + (n+m)^2-n-m-4 \\
        &\quad\geq (n+m)^2-n-m-4 \\
        &\quad> 0.
    \end{align*}
    It follows that the 1-norm relaxation $\left\|c^{\VV - \RR\UU\RR^T + \RR\WW - \WW^T\RR^T} + c^{\RR\UU\RR^T - \VV + \WW^T\RR^T - \RR\WW}\right\|_1$ monotonically decreases with $\dis \RR$.
\end{proof}


\polytopegeometry*
\begin{proof}
    The bi-mapping polytope $\SSS$ is defined by the $2nm$ inequalities of the form $S_{ij} \geq 0$, $n^2+m^2$ equality constraints of the form $S_{ij} = 0$, and the requirement of unit row sums $\SS\bf{1} = \bf{1}$. In particular, the facets of $\SSS$ lie in the hyperplanes $S_{ij} = 0$ for $(i,j) \in \left(\{1,\ldots,n\}\times \{1,\ldots,m\}\right) \cup \left(\{n+1,\ldots,n+m\}\times \{m+1,\ldots,m+n\}\right)$. Because every face is the intersection of a set of facets, $\Phi$ is given by the collective indices of zero entries describing the corresponding hyperplanes.
\end{proof}

\solutiononinterior*
\begin{proof}
The statement trivially holds for $\SS^* \in \RRR$ (recall that a vertex is its own interior). Otherwise, $\exists h,k,l$ s.t. $0 < S^*_{hk}, S^*_{hl} < 1$, which means that $\SS^*$ lies on the open segment whose endpoints $\SS', \SS''$ are given by $S'_{ij} = \begin{cases}
    S^*_{hk} + S^*_{hl} & \text{if $(i,j) = (h,k)$} \\
    0 & \text{if $(i,j) = (h,l)$}\\
    S^*_{ij} & \text{otherwise}
\end{cases}$, $S''_{ij} =
\begin{cases}
    0 & \text{if $(i,j) = (h,k)$} \\
    S^*_{hk} + S^*_{hl} & \text{if $(i,j) = (h,l)$}\\
    S^*_{ij} & \text{otherwise}
\end{cases}$. By applying the trace trick, we derive the following identity:
\begin{align*}
    \left\langle \SS', \grad \obj(\SS'') \right\rangle &= 2\left\langle \SS', c^\VV\SS'' c^{-\UU} + c^{-\VV}\SS''c^\UU + \big(c^\WW\SS'' c^{-\WW} + c^{-\WW}\SS'' c^\WW\big)^T \right\rangle \\
    &= 2\left\langle \SS'', c^\VV\SS' c^{-\UU} + c^{-\VV}\SS'c^\UU + \big(c^{-\WW}\SS' c^{\WW} + c^{\WW}\SS' c^{-\WW}\big)^T \right\rangle \\
    &= \left\langle \SS'', \grad \obj(\SS') \right\rangle.
\end{align*}
Denoting $\alpha \defeq \frac{S^*_{hk}}{S^*_{hk} + S^*_{hl}} \in (0, 1)$, we get
\begin{align*}
    \obj(\SS^*) &= \obj\left(\alpha \SS' + (1-\alpha)\SS''\right) \\
    &= \Big\langle \alpha \SS' + (1-\alpha)\SS'', \frac{1}{2}\grad\obj\left(\alpha \SS' + (1-\alpha)\SS''\right)\Big\rangle \\
    &= \alpha^2\obj(\SS') + \alpha(1-\alpha)\left\langle \SS'', \grad \obj(\SS') \right\rangle + (1-\alpha)^2\obj(\SS'') \\
    &\leq (1 - 2\alpha(1-\alpha)\obj(\SS^*) + \alpha(1-\alpha)\left\langle \SS'', \grad \obj(\SS') \right\rangle,
\end{align*} which implies $$\left\langle \SS'', \grad \obj(\SS') \right\rangle \leq 2 \obj(\SS^*) \leq \left\langle \SS', \grad \obj(\SS')\right\rangle$$ and therefore
\begin{equation*}
    \grad\obj(\SS')_{hl} \leq \grad\obj(\SS')_{hk}. \tag{1}
\end{equation*}
By the analogous reasoning using $\left\langle \SS', \grad \obj(\SS'') \right\rangle$ in place of $\left\langle \SS'', \grad \obj(\SS') \right\rangle$,
\begin{equation*}
    \grad\obj(\SS'')_{hk} \leq \grad\obj(\SS'')_{hl}.    \tag{2}
\end{equation*}
At the same time, 

\begin{align*}
    \grad\obj(\SS')_{ij} - \grad\obj(\SS'')_{ij} &= \grad\obj(\SS' - \SS'')_{ij} \\
    &= 2\big(c^\VV(\SS'-\SS'')c^{-\UU} + c^{-\VV}(\SS'-\SS'')c^\UU\big)_{ij} \\
    &\hspace{2.5cm}+ \big(c^\WW(\SS'-\SS'')c^{-\WW} + c^{-\WW}(\SS'-\SS'')c^\WW\big)_{ji} \\
    &= 2\begin{aligned}[t](S^*_{hk} + S^*_{hl})&\big(c^{V_{ih}}(c^{-U_{kj}} - c^{-U_{lj}}) + c^{-V_{ih}}(c^{U_{kj}} - c^{U_{lj}}) \\ &+ c^{W_{jh}}(c^{-W_{ki}} - c^{-W_{li}}) + c^{-W_{jh}}(c^{W_{ki}} - c^{W_{li}})\big),
    \end{aligned}
\end{align*}
and in particular
\begin{equation*}
\begin{aligned}
    \grad\obj(\SS')_{hk} - \grad\obj(\SS'')_{hk} &= 2(S^*_{hk} + S^*_{hl})\big(2 - c^{-U_{lk}} - c^{U_{lk}} + 2 - c^{W_{kh}}c^{-W_{lh}} - c^{-W_{kh}}c^{W_{lh}}\big) \\ &\leq 0,
\end{aligned} \tag{3}
\end{equation*}
\begin{equation*}
\begin{aligned}
        \grad\obj(\SS')_{hl} - \grad\obj(\SS'')_{hl} &= 2(S^*_{hk} + S^*_{hl})\big(c^{-U_{kl}} + c^{U_{kl}} - 2 + c^{W_{lh}}c^{-W_{kh}} + c^{-W_{lh}}c^{W_{kh}} - 2\big) \\ &\geq 0. 
\end{aligned}\tag{4}
\end{equation*}
Combining inequalities (1), (2), (3), and (4) gives $\grad\obj(\SS')_{hk} = \grad\obj(\SS')_{hl}$. Then $$\left\langle \SS'', \grad \obj(\SS')\right\rangle = \left\langle \SS', \grad \obj(\SS') \right\rangle = 2 \obj(\SS^*)$$ and therefore $\SS' \in \argmin_{\SS \in \SSS}\obj(\SS)$.

We showed that for any $(h, k)$ s.t. $0 < S^*_{hk} < 1$ there exists a solution $\SS' \in \SSS$ s.t. $S'_{ij} = 0$ whenever $S^*_{ij}=0$ or $(i,j)=(h,k)$. By Lemma \ref{lem:polytope_geometry}, it means that every facet of $\Phi$ contains a solution on its interior. Recursively repeating this argument establishes the existence of a solution in the interior of every non-empty face of $\Phi$, which also implies that every vertex of $\Phi$ is a solution.

Since $\obj$ is quadratic, attaining the same value at any three points on the same line renders the function constant on the entire line. In particular, any face $\Psi$ is included in $\argmin_{\SS\in\SSS}\obj(\SS)$ provided that $$\partial\Psi \cup \{\SS^\mathrm{o}\} \subseteq \argmin_{\SS\in\SSS}\obj(\SS)$$ for some interior point $\SS^\mathrm{o} \in \Psi\setminus\partial\Psi$. It follows that any positive-dimension face of $\Phi$ must be a part of the solution set if all the faces of $\Phi$ of dimension smaller by 1 are. Starting from the vertices of $\Phi$ and applying induction then yields $\Phi \subseteq \argmin_{\SS\in\SSS}\obj(\SS)$.
\end{proof}

\voronoiface*
\begin{proof}
    Note that $$\|\SS - \RR\|_2^2 = \|\SS\|_2^2 - 2\langle \SS, \RR \rangle + n+m\quad\forall\RR \in \RRR$$ and therefore $$\proj_\RRR \SS \in \argmin_{\RR \in \RRR}\,\|\SS - \RR\|_2 = \argmax_{\RR \in \RRR} \;\langle \SS, \RR \rangle.$$
    Let $s_i$ denote the maximum entry in the $i$-th row of $\SS$ for $i=1,\ldots,n+m$. Then $$\max_{\RR \in \RRR} \;\langle \SS, \RR \rangle = \sum_{i=1}^n s_i,$$ which means that any $\RR^* \in \argmax_{\RR \in \RRR} \;\langle \SS, \RR \rangle$ can have $R^*_{ij} \neq 0$ only if $S_{ij} = s_i$. Because $\SS$ is a convex combination of the vertices of $\Phi$, $S_{ij} = s_i > 0$ in turn implies the existence of $\RR \in \Phi \cap \RRR$ s.t. $R_{ij} \neq 0$. By Lemma \ref{lem:polytope_geometry}, the index of every non-zero entry of $\RR^*$ is not contained in the index set of zero entries characterizing $\Phi$, and therefore $\RR^* \in \Phi$.
\end{proof}



\oneminimizer*
\begin{proof}
    $\left|\argmin_{\SS\in\SSS} \left\langle \SS, \grad\obj(\SS^*) \right\rangle\right| = 1$ if and only if the smallest entry in each row of $\grad\obj(\SS^*)$ is unique.  In order for $\grad\obj(\SS^*)_{hk} = \grad\obj(\SS^*)_{hl}$ to hold for some $h$ and $k\neq l$, the realizations $d_1, \ldots, d_\frac{n(n-1)}{2}$ of the distances in $X$ must satisfy
    $$\sum_{i=1}^{n+m}\sum_{j=1}^{n+m} \left(c^{U_{hi}}S_{ij}c^{-V_{jk}} + c^{-U_{hi}}S_{ij}c^{V_{jk}} +c^{-W_{kj}}S_{ji}c^{W_{ih}} + c^{W_{kj}}S_{ji}c^{-W_{ih}}\right) = 0,$$
    which can be rewritten as $$\sum_{i=0}^\frac{n(n-1)}{2}\sum_{j=0}^\frac{n(n-1)}{2}a_{ij}c^{d_i - d_j} = 0$$ for some $a_{ij} \in \R$ and $d_0 \defeq 0$. Because the left-hand side can be cast as a generalized polynomial through a change of variables, it must have a finite number of solutions. The probability of the distances in $X$ to form a particular solution $d_1^*, \ldots, d_\frac{n(n-1)}{2}^*$ is
    \begin{align*}
        &\mathbb{P}_\XX\left[\mathrm{D}_1=d^*_1,\ldots,\mathrm{D}_\frac{n(n-1)}{2}=d_\frac{n(n-1)}{2}^*\right] \\ &\hspace{2cm} = \mathbb{P}_\XX\left[\mathrm{D}_1=d^*_1|\mathrm{D}_2=d_2^*,\ldots,\mathrm{D}_\frac{n(n-1)}{2}=d_\frac{n(n-1)}{2}^*\right]\mathbb{P}_\XX\left[\mathrm{D}_2=d_2^*,\ldots,\mathrm{D}_\frac{n(n-1)}{2}=d_\frac{n(n-1)}{2}^*\right] \\
        &\hspace{2cm} = 0\cdot\mathbb{P}_\XX\left[\mathrm{D}_2=d_2^*,\ldots,\mathrm{D}_\frac{n(n-1)}{2}=d_\frac{n(n-1)}{2}^*\right]
        \\
        &\hspace{2cm} = 0,
    \end{align*} and therefore
    $\mathbb{P}_\XX\left[\grad\obj(\SS^*)_{hk} = \grad\obj(\SS^*)_{hl}\right] = 0.$
    Then
    \begin{align*}
        \mathbb{P}_\XX\left[\Big|\argmin_{\SS\in\SSS} \langle \SS, \grad\obj(\SS^*) \rangle\Big| > 1\right] &\leq \sum_{h=1}^{n+m}\sum_{k=1}^{n+m}\sum_{l=k+1}^{n+m} \mathbb{P}_\XX\left[\grad\obj(\SS^*)_{hk} = \grad\obj(\SS^*)_{hl}\right] \\ &= 0.
    \end{align*}
\end{proof}

\nonconvexity*
\begin{proof}
    Let $\lmax = \lambda_1 \geq \ldots \geq \lambda_{(n+m)^2}$ denote the eigenvalues of $\HH$. We already established using the Courant-Fischer theorem that
    \begin{align*}
        \lambda_1 &\geq \frac{\|\HH\|_1}{(n+m)^2} \\
        &= \frac{2}{(n+m)^2}(\big\|c^\UU\big\|_1\big\|c^{-\VV}\big\|_1 + \big\|c^{-\UU}\big\|_1\big\|c^{\VV}\big\|_1 + \big\|c^\WW\big\|_1\big\|c^{-\WW^T}\big\|_1 + \big\|c^{-\WW}\big\|_1\big\|c^{\WW^T}\big\|_1) \\
        &= \frac{8}{(n+m)^2}\big\|c^\WW\big\|_1\big\|c^{-\WW}\big\|_1.
    \end{align*}
    To bound $\lambda_{-}\defeq\sum_{\lambda_i < 0}|\lambda_i|$ from above
    , we will first introduce two new notations for convenience. For any two matrices $\mathbf{A}$ and $\mathbf{B}$, let $\mathbf{A} \ominus \mathbf{B}$ denote a matrix operation analogous to the Kronecker product but with subtraction in place of multiplication. In addition, define $\ch:\R\to\R$ as $\ch(a) \defeq c^a + c^{-a}$ and its entrywise counterpart for matrices $\ch:\R^{p\times q}\to\R^{p\times q}$, so that we are able to compactly write
    \begin{align*}
    \frac{1}{2}\HH &= \ch(\UU\ominus\VV) + \ch(\WW\ominus\WW^T)\KK \\
    &= \ch(\UU\ominus\VV) + \ch\big((\WW\ominus\WW^T)\KK\big).
    \end{align*}
    Note that $\big\langle \ch(\mathbf{A}), \ch(\mathbf{B}) \big\rangle = \big\|\ch(\mathbf{A}+\mathbf{B})\big\|_1 + \big\|\ch(\mathbf{A}-\mathbf{B})\big\|_1$. Furthermore, if $\|\mathbf{A}\|_1$ and $\|\mathbf{A}\|_\infty$ are fixed, $\big\|\ch(\mathbf{A})\big\|_1$ is maximized by the highest possible count of entries of $\mathbf{A}$ equal to $\big\|\mathbf{A}\big\|_\infty$ due to the superadditivity of $\ch$ (see also a proof based on Lagrange multipliers under ``\href{https://math.stackexchange.com/questions/1355638/upper-bound-on-sum-of-exponential-functions}{upper bound on sum of exponential functions}'' on Mathematics Stack Exchange). As a consequence, for $\mathbf{A} \in \R^{p\times q}$
    \begin{align*}
        \big\|\ch(\mathbf{A})\big\|_1 &\leq \left\lceil\frac{\|\mathbf{A}\|_1}{\|\mathbf{A}\|_\infty}\right\rceil\ch\big(\|\mathbf{A}\|_\infty\big) + \left\lfloor pq - \frac{\|\mathbf{A}\|_1}{\|\mathbf{A}\|_\infty}\right\rfloor \ch(0)\\
        &\leq \left(\frac{\|\mathbf{A}\|_1}{\|\mathbf{A}\|_\infty} + 1\right)\left(\ch\big(\|\mathbf{A}\|_\infty\big) - 2\right) + 2pq \\
        &\leq \left(\sqrt{pq}\frac{\|\mathbf{A}\|_2}{\|\mathbf{A}\|_\infty} + 1\right)\left(\ch\big(\|\mathbf{A}\|_\infty\big) - 2\right) + 2pq.
    \end{align*}
    The looser bound in terms of the Frobenius norm allows for its tractable computation when $\mathbf{A} = \mathbf{B}\ominus\mathbf{C}$ for some $\mathbf{B} \in \R^{p\times q}, \mathbf{C} \in \R^{r \times s}$ with non-negative entries, as then
    $$\|\mathbf{A}\|_2^2 = \|\mathbf{B}\ominus\mathbf{C}\|_2^2 = pq\|\mathbf{B}\|_2^2 + rs\|\mathbf{C}\|_2^2 - 2\|\mathbf{B}\|_1\|\mathbf{C}\|_1.$$
    
    Next, we bound the eigenvalues' total magnitude from above. Using the symmetry of $\HH$,
    \begin{align*}
        \hspace{-1.5cm}\frac{1}{4}\sum_{i=1}^{(n+m)^2}\lambda_i^2 &= \left\|\frac{1}{2}\HH\right\|_2^2 \\
        &= \big\|\ch(\UU\ominus\VV)\big\|_2^2 + \big\|\ch(\WW\ominus\WW^T)\big\|_2^2 + 2\Big\langle \ch(\UU\ominus\VV), \ch\big((\WW\ominus\WW^T)\KK\big)\Big\rangle \\
        &= \big\|\ch(2\UU\ominus2\VV)\big\|_1 + \big\|\ch(2\WW\ominus2\WW^T)\big\|_1 + 4(n+m)^4
        \\&\hspace{1cm}+ 
        2\Big\|\ch\big(\UU\ominus\VV + (\WW\ominus\WW^T)\KK\big)\Big\|_1 +2\Big\|\ch\big(\UU\ominus\VV - (\WW\ominus\WW^T)\KK\big)\Big\|_  1 \\
        &\leq \Bigg((n+m)^2\frac{\|\UU\ominus\VV\|_2 + \|\WW\ominus\WW^T\|_2 + \big\|\UU\ominus\VV + (\WW\ominus\WW^T)\KK\big\|_2 + \big\|\UU\ominus\VV - (\WW\ominus\WW^T)\KK\big\|_2}{\dmax}
        \\
        &\hspace{1cm}+6\Bigg)\big(\ch(2\dmax) - 2\big) + 16(n+m)^4.
    \end{align*}
    From
    \begin{align*}
        \big\|\UU\ominus\VV \pm (\WW\ominus\WW^T)\KK\big\|_2 &= \sqrt{\|\UU\ominus\VV\|_2^2 + \|\WW\ominus\WW^T\|_2^2 \pm 2\big\langle(\WW\ominus\WW^T)\KK, \UU\ominus\VV\big\rangle}
    \end{align*}
    and \begin{align*}
        \big\langle(\WW\ominus\WW^T)\KK, \UU\ominus\VV\big\rangle &= \sum_{i,j,h,k=1}^{n+m}(W_{ik} - W_{jh})(U_{ij} -V_{hk}) \\ &=\sum_{i,j,h,k=1}^{n+m}(W_{ik}U_{ij} - W_{jh}U_{ji} - W_{ik}V_{hk} + W_{jh}V_{kh}) \\
        &= (n+m)\sum_{i,j,k=1}^{n+m}(W_{ij}U_{ik} - W_{ij}U_{ik} + W_{ji}V_{ki} - W_{ji}V_{ki}) \\
        &= 0,
    \end{align*}
    it follows that
    \begin{align*}
        \frac{1}{4}\sum_{i=1}^{(n+m)^2}\lambda_i^2
        &\leq \Bigg((n+m)^2\frac{\|\UU\ominus\VV\|_2 + \|\WW\ominus\WW^T\|_2 + 2\sqrt{\|\UU\ominus\VV\|_2^2 + \|\WW\ominus\WW^T\|_2^2}}{\dmax}
        \\
        &\hspace{1cm}+6\Bigg)\big(\ch(2\dmax) - 2\big) + 16(n+m)^4 \\
        &= \left((n+m)^2\frac{(2\sqrt{2} + 4)\sqrt{(n+m)^2\|\WW\|_2^2 - \|\WW\|_1^2}}{\dmax}+6\right)\big(\ch(2\dmax) - 2\big) + 16(n+m)^4 \\
        &= \pmax\big(\ch(2\dmax) - 2\big) + 16(n+m)^4
    \end{align*}
    and therefore
    \begin{align*}
        \lambda_- &\leq \sum_{i=2}^{(n+m)^2}|\lambda_i| \\
        &\leq \sqrt{\big((n+m)^2 - 1\big)\left(-\lambda_1^2 + \sum_{i=1}^{(n+m)^2}\lambda_i^2\right)} \\
        &\leq 2(n+m)\sqrt{16(n+m)^4 + \pmax\big(\ch(2\dmax) -2\big) - \frac{16}{(n+m)^4}\big\|c^\WW\big\|_1^2\big\|c^{-\WW}\big\|_1^2} \\
        &= 4(n+m)^2\frac{\alpha}{\frac{1}{2} - \alpha}.
    \end{align*}
    Finally,
    \begin{align*}
        \nconv(\sigma) &= \frac{\lambda_-}{2\lambda_- + 8(n+m)^2} \\
        &= \frac{1}{2} - \frac{2(n+m)^2}{\lambda_- + 4(n+m)^2} \\
        &\leq \frac{1}{2} - \frac{1}{\frac{2\alpha}{\nicefrac{1}{2}-\alpha} + 2} \\
        &= \alpha.
    \end{align*}
\end{proof}

\end{appendices}

\bibliographystyle{alpha}
\bibliography{references/references.bib}

\end{document}